\documentclass[final,3p,times,twocolumn]{elsarticle}
\usepackage{amssymb,amsmath,amsthm,bm}
\usepackage{geometry}
\usepackage{algorithm}
\usepackage{algpseudocode}
\usepackage{graphicx}
\usepackage[retainorgcmds]{IEEEtrantools}
\graphicspath{{figures-pdf/}}
\usepackage[normalem]{ulem}
\usepackage{url}
\usepackage{lineno}
\usepackage{bm}
\usepackage{graphicx,color,overpic}

\usepackage{float}

\newlength\imagewidth
\newlength\figwidth
\newlength\sfigwidth
\newlength\vfigskip
\renewcommand{\vec}[1]{\mathbf{#1}}

\usepackage{color}
\definecolor{dgreen}{rgb}{0,.6,0}
\newtheorem{property}{Property}
\newtheorem{corollary}{Corollary}

\usepackage[font=small,skip=0pt]{caption}
\begin{document}

\begin{frontmatter}

%\title{Breaking an Image Scrambling Encryption Algorithm of Pixel Bit Based on Chaos Map}
%\thanks{A preliminary version of this paper .}

\title{On the cryptanalysis of Fridrich's chaotic image encryption scheme}

\author[cn-xtu-cie]{Eric Yong Xie}

\author[cn-xtu-cie]{Chengqing Li\corref{corr}}
\ead{DrChengqingLi@gmail.com}

\author[cn-gdu]{Simin Yu}

\author[cn-cas]{Jinhu L\"u}

\cortext[corr]{Corresponding author.}

\address[cn-xtu-cie]{
Hunan Province Cooperative Innovation Center for Wind Power Equipment and Energy Conversion,\\
College of Information Engineering, Xiangtan University, Xiangtan 411105, Hunan, China}
%\address{MOE (Ministry of Education) Key Laboratory of Intelligent Computing and Information Processing,\\
%College of Information Engineering, Xiangtan University, Xiangtan 411105, Hunan, China}
%\address[cn-xtu-icip]{, Xiangtan University, China}

\address[cn-gdu]{College of Automation, Guangdong University of Technology, Guangzhou 510006, Guangdong, China}

\address[cn-cas]{Academy of Mathematics and Systems Sciences, Chinese Academy of Sciences, Beijing 100190, China}

\begin{abstract}
Utilizing complex dynamics of chaotic maps and systems in encryption was studied comprehensively in the past two and a half decades.
In 1989, Fridrich's chaotic image encryption scheme was designed by iterating chaotic position permutation and value substitution some rounds,
which received intensive attention in the field of chaos-based cryptography. In 2010, Solak \textit{et al.} proposed a chosen-ciphertext attack on
the Fridrich's scheme utilizing influence network between cipher-pixels and the corresponding plain-pixels.
Based on their creative work, this paper scrutinized some properties of Fridrich's scheme with concise mathematical language.
Then, some minor defects of the real performance of Solak's attack method were given. The work provides
some bases for further optimizing attack on the Fridrich's scheme and its variants.
\end{abstract}
\end{frontmatter}

\section{Introduction}

The complex dynamics of chaotic systems attracts researchers to utilize them as a new way to
design secure and efficient encryption schemes \cite{Xiangtao:Chaos:2007,LiuH:cat:OLT14,SYu:ARM:CASVT15,YCZhou:Chaotic:TC2015,Hua2016IS,Latif:SPN:SP16}. The first chaos-based encryption scheme was proposed in 1989 \cite{Matthews:derivation:Cryptologia89},
where a chaotic equation
\begin{equation}
g(x)=(\beta+1)(1+1/\beta)^\beta x(1-x)^\beta, \quad \beta\in[1, 4]
\label{eq:derivation}
\end{equation}
was derived to generate pseudo-random number sequence and then mask the plaintext with modulo addition.
Soon after publication of \cite{Matthews:derivation:Cryptologia89}, it was pointed out that period of the sequence
generated by iterating Eq.~(\ref{eq:derivation}) may be very short, especially when it is implemented with small computing precision,
which may seriously compromise the security level of the scheme \cite{Wheeler:Problems:Cryptologia89}.
Some special defects and properties of chaotic systems may facilitate cryptanalysis of chaos-based encryption schemes, e.g. chaotic synchronization \cite{Beth:Chaos:Crypto94}, chaotic ergodicity \cite{Arroyo:ergodicity:IJMP:2009},
and parameter identification of chaotic system \cite{Solak:identification:CAS04}.
The inadequate combination of chaotic dynamics and encryption architectures makes the complexity of recovering its secret key from some pairs of plain-texts and the corresponding cipher-texts,
encrypted with the same secret key, lower than that of brute-force attack \cite{Biham:Chaotic:Crypto91,Arroyo:Permutation:SP13,Li:hyperchaotic:ND2013,Yap:cryptanalysis:ND15}. Some general rules on evaluating security of chaos-based encryption schemes can be found in \cite{AlvarezLi:Rules:IJBC2006,ShujunLi:ChaosBook2011}.

As quantitatively analyzed in \cite{Lcq:Optimal:SP11,Licq:hierarchical:SP2016}, any position permutation-only encryption scheme can be efficiently broken
with only $O(\lceil\log_L(H\cdot W) \rceil)$ known/chosen plaintexts and the computational complexity of magnitude $O(H\cdot W\cdot \lceil\log_L(H\cdot W) \rceil)$,
where $L$ denotes the number of different gray-values of the plaintexts, and $H\times W$ (height$\times$width) is the size of the encryption scheme's \textit{permutation domain} , whose every element denotes the mapping relation between
the relative position of a permuted element in the plaintext and that in the corresponding ciphertext. As suggested in \cite{Shannon:Entropy:BSTJ49}, iterating position permutation and value substitution sufficient rounds
can make an encryption scheme very strong against all kinds of attacks.
Considering significant impact of the structure of Fridrich's scheme on a great number of chaotic encryption schemes, Solak's chosen-ciphertext attack method proposed in \cite{Solak:Fridrich:IJBC10}
can be considered as a breakthrough in the field of chaotic cryptanalysis.

According to the record of \textit{Web of Science}, both papers \cite{Fridrich:IJBC98} and \cite{YaobinMao:CSF2004}
have been cited more than 500 times up to Aug 2016. Inspired by using space network (function graph) for attacking hash function in \cite{Leurent:AttackMAC:ASIACRYPT2013},
we re-summarized some properties of Fridrich's chaotic image encryption scheme with the methodology of complex networks (binary matrix). Then, we further evaluated the real performance
of Solak's chosen-ciphertext attack method and found that it owns some minor defects. In addition, the performance of extension of the attack idea to Chen's scheme proposed in \cite{YaobinMao:CSF2004,Mao:3Dbakermap:IJBC04} was also briefly evaluated.

The rest of this paper is organized as follows. Section~\ref{sec:encryptscheme} concisely introduces Fridrich's chaotic image encryption scheme.
Our cryptanalytic results of the scheme are presented in Sec.~\ref{sec:cryptanalysis} in detail. The last section
concludes the paper.

\section{Fridrich's chaotic image encryption scheme}
\label{sec:encryptscheme}

The plaintext encrypted by Fridrich's chaotic image encryption scheme is a gray-scale image of size $H\times W$
\footnote{For simplicity, use $HW$ to denote $H\cdot W$.}, which can be denoted by a sequence of length $HW$ in domain $\mathbb{Z}_{256}$, $\vec{I}=[I(i)]_{i=0}^{HW-1}$,
by scanning it in the raster order. The corresponding cipher-image is $\vec{I}'=[I'(i)]_{i=0}^{HW-1}$.
The framework of Fridrich's scheme can be described as follows.
\begin{itemize}
\item \textit{Encryption procedure}:

\par 1) \emph{Position Permutation}: for $i=0\sim HW-1$, do
\begin{equation}
\vec{I}^*(w(i))=\vec{I}(i),
\label{eq:diffusion}
\end{equation}
where \textit{permutation matrix} $\vec{W}=\{w(i)\}_{i=0}^{HW-1}$ satisfies $w(i_1)\neq w(i_2)$ for any $i_1\neq i_2$.

\par 2) \emph{Value Substitution}: for $i=0\sim HW-1$, carry out substitution function
\begin{equation}
\vec{I}'(i)=\vec{I}^*(i)\boxplus g(\vec{I}'(i-1))\boxplus h(i),
\label{Frid:confusion}
\end{equation}
where $a\boxplus b=(a+b)\bmod 256$, $g$: $\mathbb{Z}_{256}\rightarrow \mathbb{Z}_{256}$
is a fixed nonlinear function, $\vec{H}=\{h(i)\}_{i=0}^{HW-1}$ is a pseudo-random number sequence, and $\vec{I}'(-1)$ is a pre-defined parameter $c$.

\par 3) \emph{Repetition}: set $\vec{I}=\vec{I}'$ and repeat the above two steps for $r-1$ times, where $r$ is a predefined
positive integer.

\item \textit{Decryption procedure}: it is similar to the encryption procedure except that the two main encryption
steps are carried out in a reverse order, the permutation matrix $\vec{W}$ is replaced by its inverse, and Eq.~(\ref{Frid:confusion}) is replaced
by equation
\begin{equation}
\vec{I}^*(i)=\vec{I}'(i)\boxminus g(\vec{I}'(i-1))\boxminus h(i),
\label{Frid:deconfusion}
\end{equation}
where $(a\boxminus b)=(a-b+256)\bmod 256$.
\end{itemize}

Incorporating Eq.~(\ref{eq:diffusion}) into Eq.~(\ref{Frid:confusion}), one can get
\begin{equation}
\vec{I}'(i)=\vec{I}(w^{-1}(i))\boxplus g(\vec{I}'(i-1))\boxplus h(i),
\label{eq:encryption1r}
\end{equation}
where $\vec{W}^{-1}=\{w^{-1}(i)\}_{i=0}^{HW-1}$ is the inverse of $\vec{W}$.
Combining Eq.~(\ref{eq:diffusion}) and Eq.~(\ref{Frid:deconfusion}), one has
\begin{equation}
\vec{I}(i)=\vec{I}'(w(i))\boxminus g(\vec{I}'(w(i)-1))\boxminus h(w(i)).
\label{eq:decryption1r}
\end{equation}
Since publication of \cite{Fridrich:IJBC98}, a great number of methods have been proposed to modify some elements of the framework of
Fridrich's scheme from various aspects, such as using novel methods to generate the permutation matrix; defining new concrete function $g$ in Eq.~(\ref{Frid:confusion}); changing the involved operations
in the substitution function.

\section{Cryptanalysis}
\label{sec:cryptanalysis}

To facilitate further security analysis of Fridrich's scheme, we re-present some properties of Fridrich's scheme reported in \cite[Sec. 3]{Solak:Fridrich:IJBC10}
with the methodology of matrix theory. Some critical details are appended to make their description complete, especially the conditions in Property~\ref{property:w0}
were found by us in the experiments.

\subsection{Some properties of Fridrich's scheme}

\begin{property}
There exists \textit{influence path} between the $i$-th pixel of $\vec{I}$ and the $j$-th one of $\vec{I}'$ (the value of the former may be influenced by that of the latter)
if and only if
\begin{equation*}
(\widehat{\mathbf{T}})^r(i,j)>0,
\end{equation*}
where $\widehat{\mathbf{T}}=\mathbf{P}\cdot \mathbf{T}$,
\begin{equation*}
\mathbf{P}(i, j)=
\begin{cases}
1 & \mbox{if } j=w(i);\\
0 & \mbox{otherwise},
\end{cases}
\end{equation*}
and
\begin{equation*}
\mathbf{T}=
\left(\begin{array}{ccccc}
1  &    &   &   &    \\
1  & 1  &   &   \multicolumn{2}{c}{\raisebox{1ex}[0pt]{\huge \textrm{0}}}     \\
   & 1  & 1 &   &   \\
\multicolumn{2}{c}{\raisebox{-1ex}[0pt]{\huge \textrm{0}}} &  \ddots & \ddots &  \\
   &    &   & 1 &  1
\end{array}\right)_{HW\times HW}.
\end{equation*}
\end{property}
\begin{proof}
First, we consider the case when $r$ is equal to one. Observing Eq.~(\ref{Frid:deconfusion}), one can see that relation between $\vec{I}^*$ and $\vec{I}'$
can be presented by the matrix $\mathbf{T}$: the value of the $i$-th pixel of $\vec{I}^*$ is influenced by that of the $j$-th one of $\vec{I}'$ if $\mathbf{T}(i,j)>0$
and not otherwise, where
\begin{equation}
\mathbf{T}(i, j)=
\begin{cases}
1 & \mbox{if } 0\leq i-j\leq 1;\\
0 & \mbox{otherwise}.
\end{cases}
\label{eq:matrixT}
\end{equation}
Permutation operation in Eq.~(\ref{eq:diffusion}) can be presented as multiplication of the permuted vector and an elementary matrix:
\begin{equation}
\vec{I}^*=\vec{I}\cdot \mathbf{P}.
\end{equation}
So, one can assure that the value of the $i$-th pixel of $\vec{I}$ may be influenced by that of the $j$-th pixel of $\vec{I}'$ if $\widehat{\mathbf{T}}(i,j)>0$
and not otherwise. Note that the influence may be cancelled by the modulo operation in Eq.~(\ref{eq:decryption1r}).
If $r>1$, one can easily derive that the value of $i$-th pixel of $\vec{I}$ is influenced by that of the $j$-th one of $\vec{I}'$ if and only if
\begin{equation*}
\left(\widehat{\mathbf{T}}\right)^r(i,j)>0.
\end{equation*}
\end{proof}

\begin{property}
If $w(x)+1=w(y)$, $w(y)+1=w(z)$, difference of two sets of entries of $\left(\widehat{\mathbf{T}}\right)^r$ is a subset of another similar set:
\begin{multline}
\{j\,|\,(\widehat{\mathbf{T}})^r(y, j)>0\} \backslash \{j\,|\,(\widehat{\mathbf{T}})^r(x, j)>0\}\subset \\
\{j\,|\,(\widehat{\mathbf{T}})^r(z, j)>0\},
\label{eq:condtionxi}
\end{multline}
where $x, y, z\in \mathbb{Z}_{HW}$.
\end{property}
\begin{proof}
From the definition of matrix multiplication and matrix (\ref{eq:matrixT}), one has
\begin{IEEEeqnarray}{rCl} 	
\IEEEeqnarraymulticol{3}{l}{ 	
(\widehat{\mathbf{T}})^r(x, j)}\nonumber\\
& = & \sum_{k=1}^{HW} \widehat{\mathbf{T}}(x , k) \cdot (\widehat{\mathbf{T}})^{r-1}(k, j) \nonumber\\
& = & \sum_{k=1}^{HW} \sum_{l=1}^{HW} \mathbf{P}(x, l) \cdot {\mathbf{T}}(l, k) \cdot (\widehat{\mathbf{T}})^{r-1}(k, j) \nonumber\\
& = & \sum_{k=1}^{HW} \mathbf{P}(x, w(x)) \cdot {\mathbf{T}}(w(x), k) \cdot (\widehat{\mathbf{T}})^{r-1}(k, j)  \nonumber\\
& = & \sum_{k=1}^{HW} \mathbf{T}(w(x), k) \cdot (\widehat{\mathbf{T}})^{r-1}(k, j)  \label{eq:multiplication}
.
\label{eq}
\end{IEEEeqnarray}
As $\mathbf{T}(w(x), k)=0$ when $k\not\in\{w(x)-1, w(x)\}$, one can get the following two points when $r>1$:
\begin{itemize}
\item $(\widehat{\mathbf{T}})^r(x, j)>0$ if and only if
$(\widehat{\mathbf{T}})^{r-1}(w(x), j)>0$ or $(\widehat{\mathbf{T}})^{r-1}(w(x)-1, j)>0$ when $w(x)\neq 0$;\\
\item $(\widehat{\mathbf{T}})^r(x, j)>0$ if and only if
$(\widehat{\mathbf{T}})^{r-1}(w(x), j)>0$ when $w(x)=0$.
\end{itemize}
This means that
\begin{IEEEeqnarray}{rCl} 	
\IEEEeqnarraymulticol{3}{l}{	
\{j\,|\,(\widehat{\mathbf{T}})^r(x, j)>0\}
}\nonumber\\\;\; 	
& = &
\begin{cases}
\{j\,|\,(\widehat{\mathbf{T}})^{r-1}(w(x), j)>0\}         & \mbox{if } w(x)=0;\\
\{j\,|\,(\widehat{\mathbf{T}})^{r-1}(w(x), j)>0\}\cup     & \\
\{j\,|\,(\widehat{\mathbf{T}})^{r-1}(w(x)-1, j)>0\}       & \mbox{otherwise}.
\end{cases}
\label{eq:T1x}
\end{IEEEeqnarray} 	

Since $w(y)\neq 0$, $w(z)\neq 0$, one can deduce the following two sets of equations similarly:
\begin{IEEEeqnarray}{rCl} 	
\IEEEeqnarraymulticol{3}{l}{ 	
\{j\,|\,(\widehat{\mathbf{T}})^r(y, j)>0\}}\nonumber\\
& = & \{j\,|\,(\widehat{\mathbf{T}})^{r-1}(w(y), j)>0\}\cup \nonumber\\
& & {\;  } \{j\,|\,(\widehat{\mathbf{T}})^{r-1}(w(y)-1, j)>0\},\nonumber\\
& = & \{j\,|\,(\widehat{\mathbf{T}})^{r-1}(w(y), j)>0\}\cup  \nonumber\\
& & {\;  } \{j\,|\,(\widehat{\mathbf{T}})^{r-1}(w(x), j)>0\}, \label{eq:T1yj}
\end{IEEEeqnarray} 	
and
\begin{IEEEeqnarray}{rCl} 	
\IEEEeqnarraymulticol{3}{l}{ 	
\{j\,|\,(\widehat{\mathbf{T}})^r(z, j)>0\}}\nonumber\\
&=&\{j\,|\,(\widehat{\mathbf{T}})^{r-1}(w(z), j)>0\}\cup \nonumber\\
& & {\;  }\{j\,|\,(\widehat{\mathbf{T}})^{r-1}(w(z)-1, j)>0\},\nonumber\\
&=&\{j\,|\,(\widehat{\mathbf{T}})^{r-1}(w(z), j)>0\}\cup \nonumber\\
& & {\;  }\{j\,|\,(\widehat{\mathbf{T}})^{r-1}(w(y), j)>0 \}.
\label{eq:T1zy}
\end{IEEEeqnarray} 	
Using relation between absolute complement and relative complement of a set, one can
obtain difference of the left parts of Eq.~(\ref{eq:T1yj}) and Eq.~(\ref{eq:T1x}),
\begin{IEEEeqnarray*}{rCl} 	
\IEEEeqnarraymulticol{3}{l}{	
\{j\,|\,(\widehat{\mathbf{T}})^r(y, j)>0\} \backslash \{j\,|\,(\widehat{\mathbf{T}})^r(x, j)>0\}
}\nonumber\\\;\; 	
& = &
\begin{cases}
\{j\,|\,(\widehat{\mathbf{T}})^{r-1}(w(y), j)>0\}  &\\
{\ }\cap\{j\,|\,(\widehat{\mathbf{T}})^r(x, j)=0\} & \mbox{if } w(x)=0;\IEEEeqnarraynumspace\\
\{j\,|\,(\widehat{\mathbf{T}})^{r-1}(w(y), j)>0\}  &\\
{\ }\cap\{j\,|\,(\widehat{\mathbf{T}})^{r-1}(w(x), j)=0\}    &\\
{\ }\cap\{j\,|\,(\widehat{\mathbf{T}})^{r-1}(w(x)-1, j)=0\}  & \mbox{otherwise}.
\end{cases}
\end{IEEEeqnarray*}
For either case of the above equation, one can get
\begin{multline}
\{j\, |\, (\widehat{\mathbf{T}})^r(y, j)>0\} \backslash \{j\, |\, (\widehat{\mathbf{T}})^r(x, j)>0\}\subset \\
\{j \,|\,(\widehat{\mathbf{T}})^r(z, j)>0\}
\end{multline}
by observing the right part of Eq.~(\ref{eq:T1zy}).
\end{proof}
\begin{corollary}
If $w(x)=0$, $w(y)=1$, then
\begin{equation}
\left\{j\,|\,(\widehat{\mathbf{T}})^r(x, j)>0\right\} \subset \left\{j\,|\,(\widehat{\mathbf{T}})^r(y, j)>0\right\}.
\label{eq:condtionx1}
\end{equation}
\end{corollary}
\begin{proof}
This corollary can be easily proofed by comparing Eq.~(\ref{eq:T1x}) and Eq.~(\ref{eq:T1yj}).
\end{proof}

\begin{property}
If $w(0)\neq 1$, $|w(x_1)-w(x_2)| \neq 1$ for any $x_1, x_2$ satisfying $|x_1-x_2|\in\{1, \cdots, 2^{r-1}-1\}$,
one has
$$w(x)=0,$$
where the $x$-th row of $(\widehat{\mathbf{T}})^r$ is its row vector containing
minimal number of non-zero element, i.e.,
\begin{equation}
|\{j\,|\,(\widehat{\mathbf{T}})^r(x, j)>0\}|=\min{\left\{ |\{j\,|\,(\widehat{\mathbf{T}})^r(i, j)>0\} |\right\}}_{i=0}^{HW-1},
\label{eq:minimal}
\end{equation}
and $r\ge 2$.
\label{property:w0}
\end{property}
\begin{proof}
As shown in Fig.~\ref{fig:pattern}, there exist and only exist three basic patterns for reducing the number of influencing cipher-pixels
for a plain-pixel. If $|w(x_1)-w(x_2)| \neq 1$ for any $x_1, x_2$ satisfying $|x_1-x_2|\in\{1, \cdots, 2^{r-1}-1\}$, the first two
patterns in Fig.~\ref{fig:pattern} can be excluded. Furthermore, the third one can be eliminated if $w(0)\neq 1$. (A concrete counter example is shown in Fig.~\ref{fig:SolakAttacks2}.)
Under the given condition in this property, there is only one influence path between any pair of cipher-pixel and plain-pixel.
So, the $x$-th row of $(\widehat{\mathbf{T}})^r$ has $2^{r-1}$ non-zero elements while other rows all have $2^{r}$ non-zero elements.
Then, one can correctly recover $w(x)=0$ by checking condition~(\ref{eq:minimal}).
\end{proof}

\begin{figure}[!htb]
\centering
\includegraphics[width=\imagewidth]{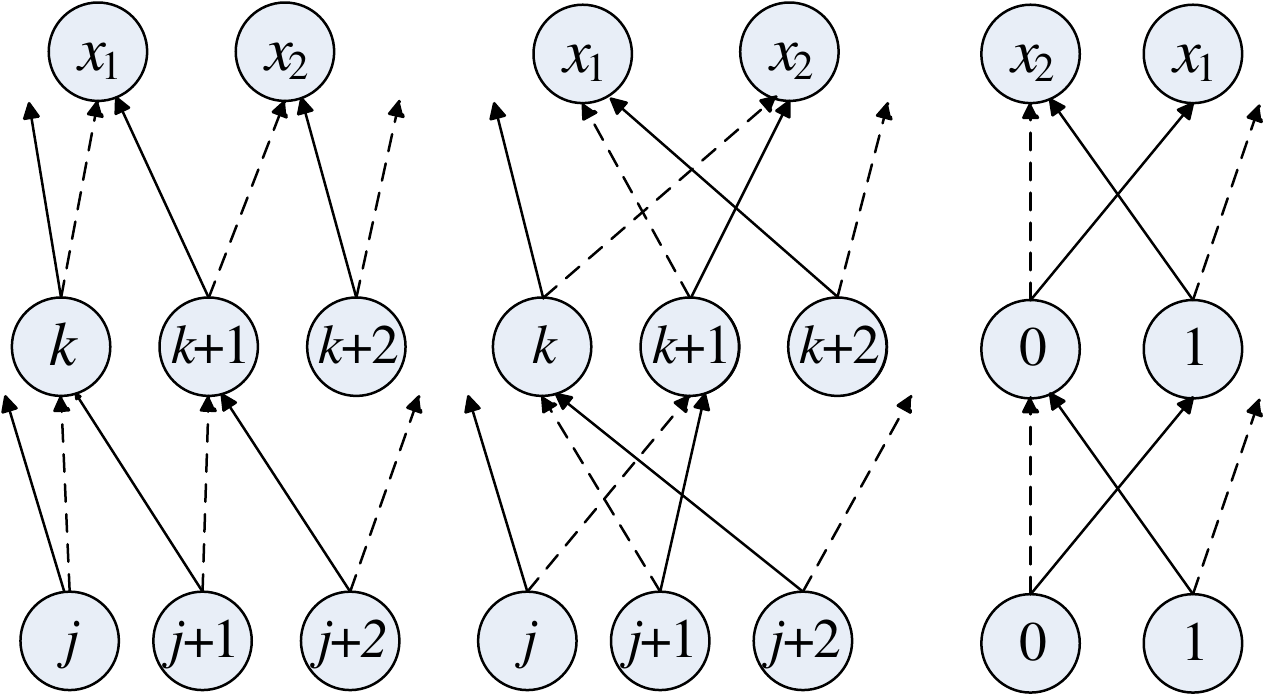}
\caption{Three basic patterns for reducing incoming influence paths of a plain-pixel, where
the solid arrow indicates influence caused by $\vec{I}'(w(i))$ and the dashed arrow denotes that caused by $\vec{I}'(w(i)-1)$,
and the symbol inside a circle denotes the index number of pixel (the same hereinafter).}
\label{fig:pattern}
\end{figure}

\begin{figure}[!htb]
\centering
\includegraphics[width=\imagewidth]{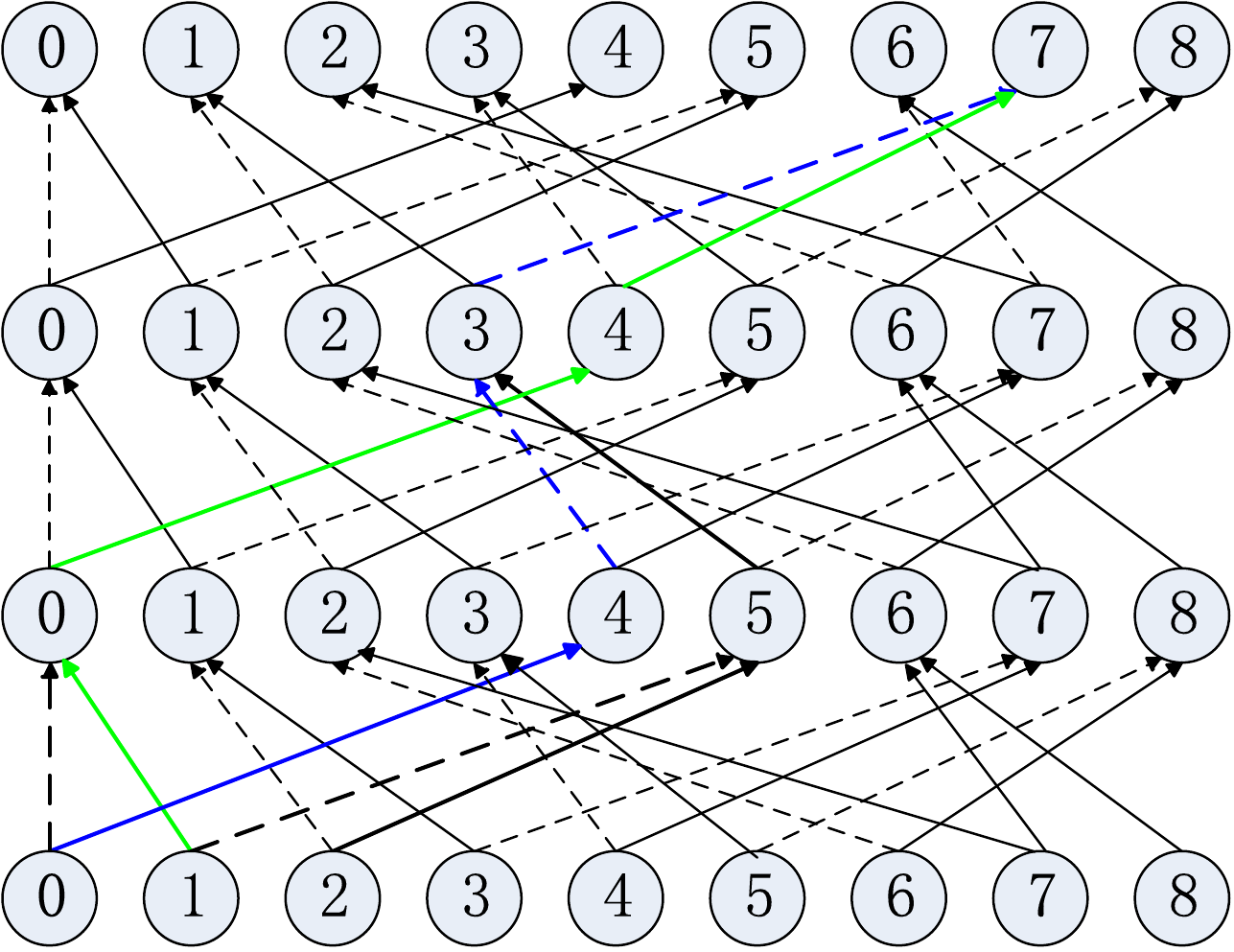}
\caption{A counterexample on Property~3, where $\vec{W}=[1, 3, 7, 5, 0, 2, 8, 4, 6]$, and $r=3$.}
\label{fig:SolakAttacks2}
\end{figure}

\subsection{Description of Solak's chosen-ciphertext attack}

To facilitate the discuss in the next subsection, we concisely describe the process of Solak's chosen-ciphertext attack method given in \cite{Solak:Fridrich:IJBC10}.
Let $\vec{\hat{W}}^{-1}=\{\hat{w}^{-1}(i)\}_{i=0}^{HW-1}$ denote the estimated version of $\vec{W}^{-1}$.
After recovering the approximate version of the influence matrix between cipher-text and the corresponding plain-text,
the Solak's attack method prunes the search space of size $HW!$ (factorial of $HW$) with the following steps:
\begin{itemize}
\item \textit{Step 1)} Set $\hat{w}^{-1}(0)=x_0$, where $x_0$ is the row number satisfying condition~(\ref{eq:minimal}) itself.

\item \textit{Step 2)} Let $\hat{w}^{-1}(1)=x_1$, where $x_1\in \vec{A}\backslash\vec{B}$,
$\vec{A}$ is the right-hand set in condition~(\ref{eq:condtionx1}) with $x=x_0$, and $\vec{B}=\{x_0\}$.

\item \textit{Step 3)} For $i=2\sim HW-1$, let $\hat{w}^{-1}(i)=x_i$, where $x_i\in \vec{A}\backslash\vec{B}$,
$\vec{A}$ is the right-hand set in condition~(\ref{eq:condtionxi}) with $x=i-1$, $y=i$, and $\vec{B}=\{x_j\}_{j=0}^{i-1}$.
If $\vec{A}$ is empty, one can assure that the current value of $\hat{w}^{-1}(i-1)$ is wrong
and process the search with its another candidate value.

\item \textit{Step 4)} Repeat the above steps iteratively till variable $i$ reaches the maximal value, $HW-1$.
\end{itemize}

%The number of influence paths of the $x_0$-th pixel can reach $2^r-\sum_{j=0}^s 2^{r-r_j}$, where $w^{r_j}(x_0)=0$, and $r_j\in \{1, \cdots, r\}$.

\subsection{Real performance of Solak's chosen-ciphertext attack}

Observing Eq.~(\ref{eq:decryption1r}), one can see that the $i_0$-th plain-pixel can be only influenced by one cipher-pixel in each encryption round
if $w(i_0)=0$. After accumulation of $r$ rounds of encryption, the number of cipher-pixels influencing the $i_0$-th plain-pixel is smaller than
that influencing other plain-pixels in a very high probability, which serves as the basis of \textit{Step 1)}. The scope of the former is $[1, 2^{r-1}]$. In contrast,
the scope of the latter is $[r+1, 2^r]$, whose lower bound can be achieved when there exist
$x$ and $t=r$ satisfying $|w(x+i+1)-w(x+i)|=1$ for $i=0\sim$ $t-1$. Observing the right part of Fig.~\ref{fig:SolakAttack_counter}, one can see that the number of cipher-pixels influencing the $x$-th plain-pixel shifts from $r+1$ to $2^r$ monotonously when $t$ is increased from $0$ to $r$. Property~\ref{property:w0} only presents an extreme condition assuring the estimation in \textit{Step 1)} is definitely right.
Interestingly, the condition in Property~\ref{property:w0} is very similar to the problem of neighbors remain neighbors after random rearrangements, discussed in \cite{Abramson:Permute:AMS1967}.
\begin{figure}[!htb]
\centering
\includegraphics[width=1.1\imagewidth]{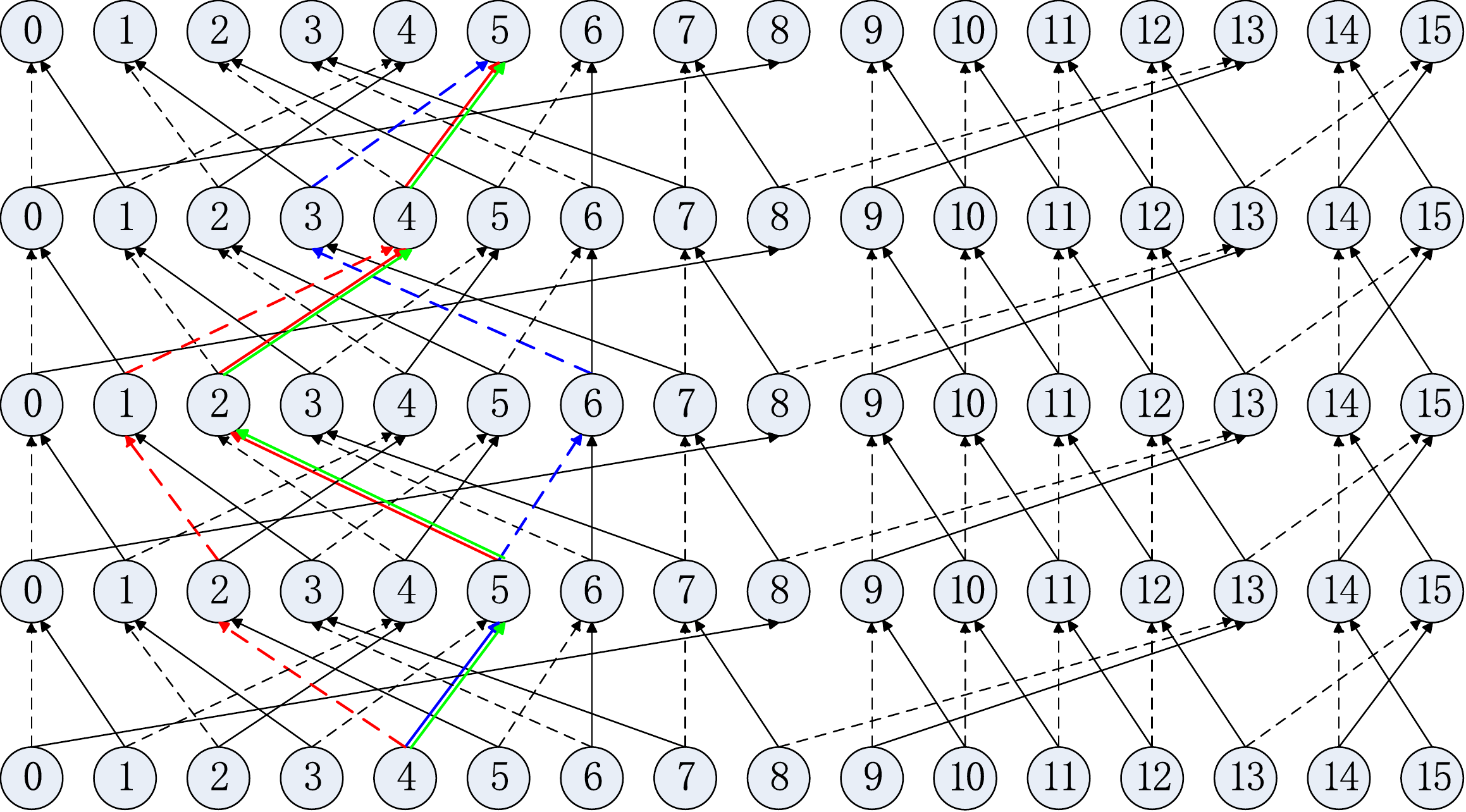}
\caption{Another counterexample about Property~3.}
\label{fig:SolakAttack_counter}
\end{figure}

Now, we give a counterexample to show deficiency of Solak's chosen-ciphertext attack method. When $\vec{W}=[1, 3, 5, 7, 2, 4, 6, 8, 0, 10, 11, 12, 13, 9, 15, 14]$, $r=$ $4$,
the influence relation between the cipher-pixels and the corresponding plain-pixels is shown in
Fig.~\ref{fig:SolakAttack_counter}. Its binary matrix form is presented in Fig.~\ref{fig:InfluenceMatrix}, which demonstrates that the 9-th row has least non-zero elements.
According to Solak's attack method, one can get $w(9)=0$, which is contradict with the real value. As the initial step has cascaded influence on the succeeding steps,
the attack is totally failed under the given secret key.

\begin{figure}[!htb]
\centering
\includegraphics[width=0.7\imagewidth]{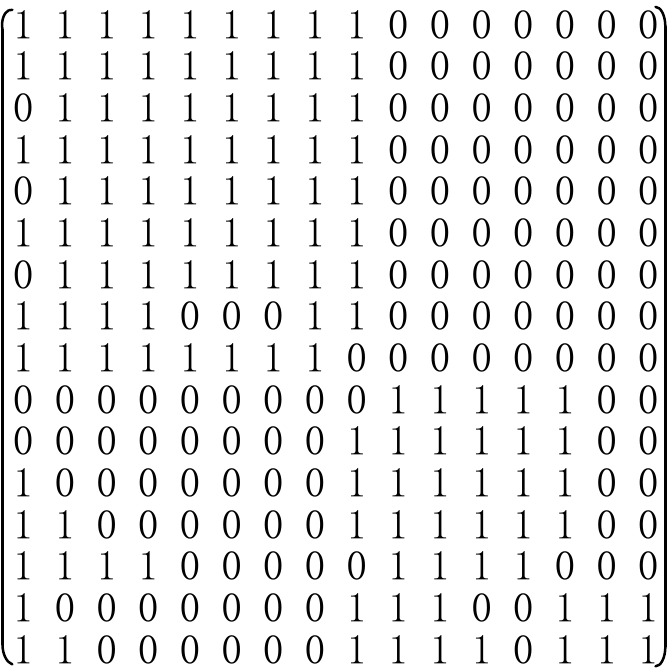}
%\caption{Obtained influence matrix when $\vec{W}=[2, 4, 6,8, 3, 5,$ $ 7, 9, 1, 11, 12, 13, 14, 10, 16, 15]$, and $r=4$.}
\caption{Obtained influence matrix when $\vec{W}=[1, 3, 5, 7, 2, 4,$ $6, 8, 0, 10, 11, 12, 13, 9, 15, 14]$, and $r=4$.}
\label{fig:InfluenceMatrix}
\end{figure}

%As shown in \cite[Sec.~4]{Solak:Fridrich:IJBC10}, when permutation vector $\vec{W}=[9, 8, 6, 12, 1, 11, 14, 15, 7, 3, 10, 2, 16, 5, 4, 13]$ and $r=3$,
As shown in \cite[Sec.~4]{Solak:Fridrich:IJBC10}, when permutation vector $\vec{W}=[8, 7, 5, 11, 0, 10, 13, 14, 6, 2, 9, 1, 15, 4, 3, 12]$ and $r=3$,
the permutation vector can be solely recovered with Solak's chosen-ciphertext attack method. However, when it is slightly changed as
%$\vec{K}=[9, 8, 6, 12, 13, 11,14, 15,7, 3, $ $10, 2, 16, 1, 4, 5]$,
$\vec{W}=[8, 7, 5, 11,12,$ $10, 13, 14, 6, 2, 9, 1, 15, 0, 3, 4]$,
eight possible values are obtained by the attack (See Fig.~\ref{fig:KeyOutput}). To verify this point further, we performed Solak's attack on a plaintext of size $1\times 32$ with
1,000 randomly assigned $\vec{W}$ and three possible encryption rounds. The following five cases, in terms of attacking results, were counted:
1) the right key is enclosed; 2) the sole result is the right key; 3) both right key and wrong key exist; 4) no any result is obtained; 5) all found results are wrong.
Let $n_1$, $n_2$, $n_3$, $n_4$ and $n_5$ denote the number of the five cases occurring among 1,000 times random experiments, respectively.
The calculated results are shown in Table~\ref{tab:random}, which demonstrates that the attack results become more worse as
$2^r$ approaches $MN$ more, which agrees with analysis in the proof of Property~\ref{property:w0}.

\begin{figure}[!htb]
\centering
\includegraphics[width=\imagewidth]{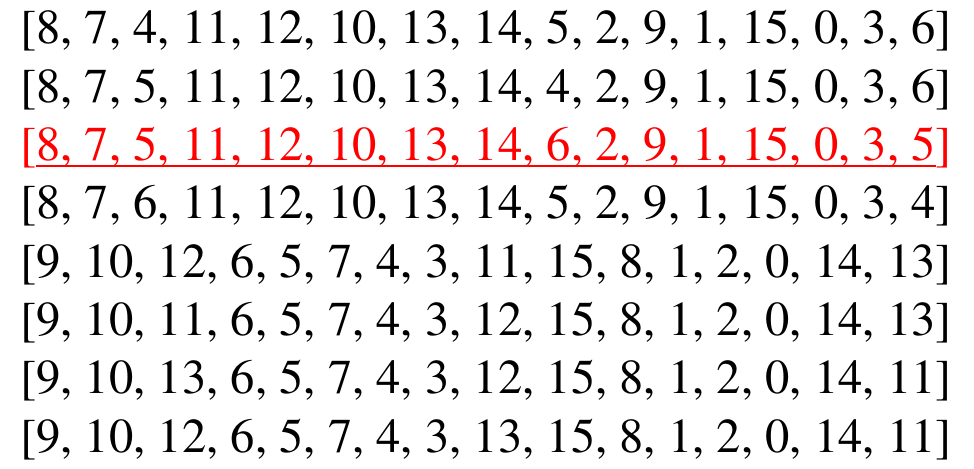}
\caption{The results of Solak's attack when
%$\vec{W}=[9, 8, 6,$ $12, 13,$ $11,$ $14, 15, 7, 3, 10, 2, 16, 1, 4, 5]$.
$\vec{W}=[8, 7, 5, 11$, $12, 10$, $13, 14$, $6, 2, 9, 1, 15, 0, 3, 4]$.
}
\label{fig:KeyOutput}
\end{figure}

\begin{table}[htb]
\centering
\caption{The number of five possible cases occurring among 1,000 random secret keys.}
\begin{tabular}{|p{1cm}|p{1cm}|p{1cm}|p{1cm}|}
\hline
$r$    &   2      &   3    &  4    \\ \hline
$n_1$  &   1000   &   957  &  814  \\ \hline
$n_2$  &   964    &   867  &  571  \\ \hline
$n_3$  &   36     &   90   &  243  \\ \hline
$n_4$  &   0      &   43   &  180  \\ \hline
$n_5$  &   0      &   0    &  6    \\ \hline
\end{tabular}
\label{tab:random}
\end{table}

%The number of root of a permutation may be very large \cite{Nicolas:Root:EJC2002,Leanos:roots:AJC2012}.
%It is influenced by all kinds of factors: the power $r$, the number of cycles, the relation between $r$ and the size of permutation.

\subsection{Real performance of extension of Solak's attack to Chen's scheme}

In \cite[Sec. 5]{Solak:Fridrich:IJBC10}, it was claimed that Solak's chosen-plaintext attack method can be applied to Chen's scheme easily and effectively
due to similar structure. However, we found some differences caused by the different basic encryption operations.

In \cite{YaobinMao:CSF2004,Mao:3Dbakermap:IJBC04}, Chen's scheme changes
Eq.~(\ref{Frid:confusion}) as
\begin{equation*}
%\label{eq:chendiffusion}
\vec{I}'(i)=h(i)\oplus[\vec{I}^*(i)\boxplus h(i)]\oplus \vec{I}'(i-1).
\end{equation*}
Accordingly, Eq.~(\ref{Frid:deconfusion}) becomes
\begin{equation}
\vec{I}^*(i)=(h(i)\oplus\vec{I}'(i)\oplus\vec{I}'(i-1))\boxminus h(i).
\label{eq:chendediffusion}
\end{equation}
Combing Eq.~(\ref{eq:diffusion}) and Eq.~(\ref{eq:chendediffusion}), one has
\begin{equation}
\label{eq:chencombine}
\vec{I}(i)=(h(w(i))\oplus\vec{I}'(w(i))\oplus\vec{I}'(w(i)-1))\boxminus h(w(i)).
\end{equation}

When $\vec{W}=[1, 3, 5, 7, 2, 4,6, 8, 0, 10, 11, 12, 13, 9, 15,$ $ 14]$, and $r=4$, the number of different influence paths between a cipher-pixel and any plain-pixel
is shown in Fig.~\ref{fig:InfluenceMatrix2}. Due to the bitwise exclusive or (XOR) operation used in Eq.~(\ref{eq:chendediffusion}), the influence of a cipher-pixel on
a influenced plain-pixel may be cancelled when there is multiple influence paths between them. So, some influence paths can not be recovered.
Note that even error of one element of the influence matrix may fail a attacking step based on the sets comparison and disable the following steps due to the
cascading influence. More discussions on composite function of modulo addition and XOR operation can be found in \cite[Sec. 3.1]{Cqli:breakmodulo:IJBC13}. Figure~\ref{fig:InfluenceMatrix_chen} shows the obtained influence matrix by changing every cipher-pixel with the fixed value one, which has six unrecovered influence paths.
As the recovered influence matrix of sufficient accuracy is requisite condition for success of the Solak's attack method, one has to improve its correct ratio
by changing cipher-pixel with the other values (more cipher-images).

\begin{figure}[!htb]
\centering
\includegraphics[width=0.7\imagewidth]{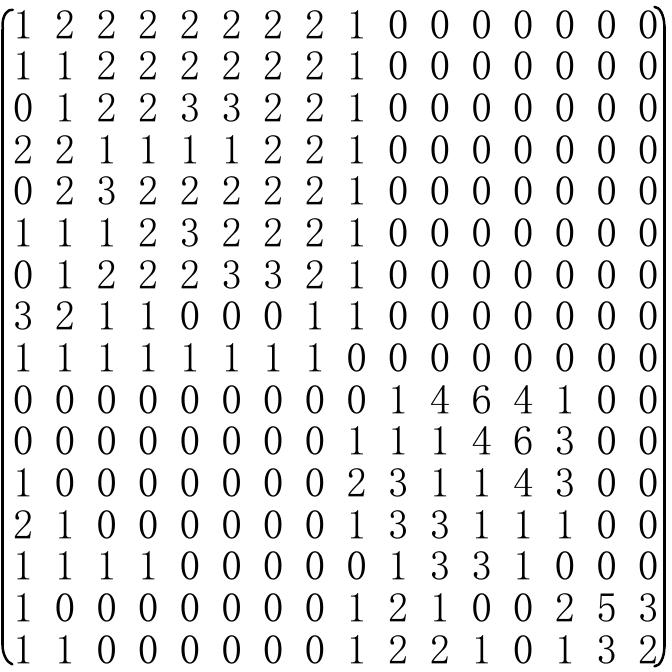}
\caption{The number of different influence paths between the $i$-th cipher-pixel and the $j$-th plain-pixel,
where $\vec{W}=[1, 3, 5, 7, 2, 4,$ $6, 8, 0, 10, 11, 12, 13, 9, 15, 14]$, and $r=4$.}
\label{fig:InfluenceMatrix2}
\end{figure}

\begin{figure}[!htb]
\centering
\includegraphics[width=0.7\imagewidth]{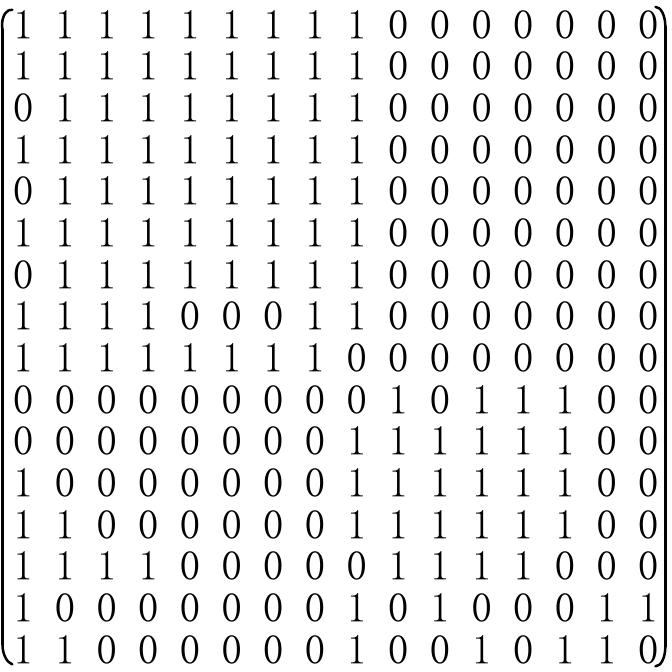}
\caption{Obtained influence matrix when $\vec{W}=[1, 3, 5, 7, 2, 4,$ $6, 8, 0, 10, 11, 12, 13, 9, 15, 14]$, and $r=4$.}
\label{fig:InfluenceMatrix_chen}
\end{figure}

Note that equivalent version of position permutation and value substitution of 1-round version Chen's scheme was successfully recovered with some chosen-plaintexts
in \cite{Kaiwang:PLA2005}. Its weak sensitivity with respect to changes of secret key and plaintext was demonstrated in detail in \cite{Li:AttackingMaoScheme2008}.

\section{Conclusion}

Based on the work in \cite{Solak:Fridrich:IJBC10}, this paper formulated some properties of Fridrich's chaotic image encryption scheme with matrix theory and
reported some minor defects of Solak's chosen-ciphertext attack method on it. The work may help designers of chaotic encryption schemes to realize fundamental
importance of the underlying encryption architecture for security performance \cite{Cqli:Logistic:IJBC15}. The following problems on cryptanalyzing Fridrich's chaotic image encryption scheme deserve further investigation: decreasing the required number of chosen-ciphertext with the special properties of the influence matrix between cipher-pixels and the corresponding plain-pixels; reducing computational complexity of the chosen-ciphertext attack;
disclosing the influence matrix under the scenario of known/chosen-plainext attack.

\section*{Acknowledgement}

This research was supported by Hunan Provincial Natural Science Foundation of China (No.~2015JJ1013),
Scientific Research Fund of Hunan Provincial Education Department~(No.~15A186), and the National Natural Science Foundation of China (No.~61532020).

\bibliographystyle{elsarticle-num}
%\bibliographystyle{IEEEtran}
%\nocite*
\bibliography{CAT_Elsevier}
\end{document}